\documentclass[onecolumn, draftclsnofoot]{IEEEtran}
\usepackage{amsthm}
\usepackage{amsmath}
\usepackage{amsfonts}
\usepackage{graphicx}
\usepackage{latexsym}
\usepackage{amssymb}
\usepackage{stmaryrd}
\usepackage{comment}
\usepackage{amscd}
\usepackage{hyperref}
\newcommand\myshade{70} 
\hypersetup{ 
	linkcolor  = red!\myshade!black,
	citecolor  = blue!\myshade!black,
	urlcolor   = blue!\myshade!black,
	colorlinks = true,
}

\usepackage{tikz}
\usetikzlibrary{
    arrows.meta,
    backgrounds,
    decorations.pathreplacing,
    fit,
    patterns,
    positioning,
    shapes.geometric
}
\tikzset{
  block/.style = {rectangle, draw, fill=blue!10, text centered, minimum height=2em, minimum width=6em},
  arrow/.style = {thick,->,>=stealth}
}

\usepackage{cite}
\usepackage{cleveref}
\usepackage{colortbl}
\usepackage[dvipsnames]{xcolor}
\usepackage[linesnumbered,ruled]{algorithm2e}
\usepackage[noend]{algpseudocode}
\allowdisplaybreaks
\usepackage{graphicx}
\usepackage{subfigure}
\usepackage{wrapfig}
\usepackage{float}
\usepackage{bm}
\usepackage{bbm}
\usepackage{enumerate}

\usepackage{mathtools}
\newcommand{\bea}{\begin{eqnarray}}
\newcommand{\eea}{\end{eqnarray}}
\newcommand{\bean}{\begin{eqnarray*}}
\newcommand{\eean}{\end{eqnarray*}}

\newcommand{\sbinom}[2]{\left( \begin{array}{c} #1 \\ #2 \end{array} \right) }

\newcommand{\cA}{{\cal A}}
\newcommand{\cB}{{\cal B}}
\newcommand{\cC}{{\cal C}}
\newcommand{\cD}{{\cal D}}

\newcommand{\cP}{{\cal P}}
\newcommand{\cQ}{{\cal Q}}
\newcommand{\cR}{{\cal R}}
\newcommand{\cS}{{\cal S}}

\newcommand{\cZ}{{\cal Z}}

\newcommand{\sG}{\script{G}}

\newcommand{\sP}{\script{P}}

\newcommand{\bfm}{{\boldsymbol m}}

\newcommand{\bfx}{{\boldsymbol x}}

\DeclareMathAlphabet{\mathbfsl}{OT1}{cmr}{bx}{it}
\newcommand{\uuu}{\kern-1pt\mathbfsl{u}\kern-0.5pt}
\newcommand{\vvv}{\kern-1pt\mathbfsl{v}\kern-0.5pt}

\newcommand{\myboxplus}{\kern1pt\mbox{\small$\boxplus$}}

\makeatletter \DeclareRobustCommand{\sbinom}{\genfrac[]\z@{}}
\makeatother
\newcommand{\G}[2]{\sbinom{{#1}\kern-1pt}{{#2}\kern-1pt}}
\newcommand{\Gq}[2]{\sbinom{{#1}\kern-0.25pt}{{#2}\kern-0.25pt}}

\newcommand{\Fqn}{\smash{{\mathbb F}_{\!q}^n}}

\newcommand{\Ps}{\smash{{\sP\kern-2.0pt}_q\kern-0.5pt(n)}}
\newcommand{\sPs}{\smash{{\sP\kern-1.5pt}_q(n)}}
\newcommand{\Ptwo}{\smash{{\sP\kern-2.0pt}_2\kern-0.5pt(n)}}
\newcommand{\Ptwom}{\smash{{\sP\kern-2.0pt}_2\kern-0.5pt(m)}}
\newcommand{\Ptwonm}{\smash{{\sP\kern-2.0pt}_2\kern-0.5pt(n+m)}}
\newcommand{\Ptwoa}{\smash{{\sP\kern-2.0pt}_2\kern-0.5pt(1)}}
\newcommand{\Ptwob}{\smash{{\sP\kern-2.0pt}_2\kern-0.5pt(2)}}
\newcommand{\Ptwoc}{\smash{{\sP\kern-2.0pt}_2\kern-0.5pt(3)}}
\newcommand{\Ptwod}{\smash{{\sP\kern-2.0pt}_2\kern-0.5pt(4)}}
\newcommand{\Ptwoe}{\smash{{\sP\kern-2.0pt}_2\kern-0.5pt(5)}}
\newcommand{\Ptwof}{\smash{{\sP\kern-2.0pt}_2\kern-0.5pt(6)}}
\newcommand{\Ptwokm}{\smash{{\sP\kern-2.0pt}_2\kern-0.5pt(2k-1)}}
\newcommand{\Pone}{\smash{{\sP\kern-2.5pt}_2\kern-0.5pt(n{-}1)}}

\newcommand{\Gr}{\smash{{\sG\kern-1.5pt}_q\kern-0.5pt(n,k)}}
\newcommand{\Gi}{\smash{{\sG\kern-1.5pt}_q\kern-0.5pt(n,i)}}
\newcommand{\Gj}{\smash{{\sG\kern-1.5pt}_q\kern-0.5pt(n,j)}}
\newcommand{\Grmk}{\smash{{\sG\kern-1.5pt}_q\kern-0.5pt(n,n-k)}}
\newcommand{\Grdk}{\smash{{\sG\kern-1.5pt}_q\kern-0.5pt(2k,k)}}
\newcommand{\Grekappa}{\smash{{\sG\kern-1.5pt}_q\kern-0.5pt(n,e+1-\kappa)}}
\newcommand{\Grtwoekappa}{\smash{{\sG\kern-1.5pt}_q\kern-0.5pt(n,2e+1-\kappa)}}
\newcommand{\Gremkappa}{\smash{{\sG\kern-1.5pt}_q\kern-0.5pt(n,e-\kappa)}}
\newcommand{\Gn}{\smash{{\sG\kern-1.5pt}_2\kern-0.5pt(n,n{-}1)}}
\newcommand{\Gnq}{\smash{{\sG\kern-1.5pt}_q\kern-0.5pt(n,n{-}1)}}
\newcommand{\Gone}{\smash{{\sG\kern-1.5pt}_2\kern-0.5pt(n,1)}}
\newcommand{\Gqone}{\smash{{\sG\kern-1.5pt}_q\kern-0.5pt(n,1)}}
\newcommand{\GTwo}{\smash{{\sG\kern-1.5pt}_2\kern-0.5pt(n,k)}}
\newcommand{\GTwonk}[2]{{\smash{{\sG\kern-1.5pt}_2\kern-0.5pt({#1},{#2})}}}
\newcommand{\Gnk}{\smash{{\sG\kern-1.5pt}_2\kern-0.5pt(n,n{-}k)}}
\newcommand{\Greone}{\smash{{\sG\kern-1.5pt}_q\kern-0.5pt(n,e{+}1)}}
\newcommand{\Gretwo}{\smash{{\sG\kern-1.5pt}_q\kern-0.5pt(n,e{+}2)}}

\newcommand{\be}[1]{\begin{equation}\label{#1}}
\newcommand{\ee}{\end{equation}}

\newcommand{\Hent}[1]{\mathbf{H}\left(#1\right)}
\newcommand{\Iinfo}[2]{\mathbf{I}\left(#1; #2\right)}
\newcommand{\condH}[2]{\mathbf{H}\left(#1 \mid #2\right)}
\newcommand{\condI}[3]{\mathbf{I}\left(#1; #2 \mid #3\right)}

\newtheorem{theorem}{Theorem}
\newtheorem{lemma}{Lemma}
\newtheorem{remark}{Remark}

\newtheorem{definition}{Definition}
\newtheorem{proposition}{Proposition}

  \newcommand*{\Initial}{I\mskip-2mu}
\newcommand*{\Final}{F\mskip-2mu}
\newcommand*{\In}{n^{\Initial}}
\newcommand*{\Fn}{n^{\Final}}
\newcommand*{\Ik}{k^{\Initial}}
\newcommand*{\Fk}{k^{\Final}}

\newcommand*{\IPart}{{\mathcal{P}}^\Initial}
\newcommand*{\FPart}{{\mathcal{P}}^\Final}
\newcommand*{\Ig}{g^{\Initial}}
\newcommand*{\Fg}{g^{\Final}}
\newcommand*{\Ir}{r^{\Initial}}
\newcommand*{\Fr}{r^{\Final}}
\newcommand*{\Imm}{\mu^{\Initial}}

\newcommand*{\Il}{\lambda^{\Initial}}

\newcommand{\bw}{\gamma_{\mathrm{R}}}

\IEEEoverridecommandlockouts
  
\begin{document}

\author{
    \IEEEauthorblockN{\textbf{Saransh~Chopra}\IEEEauthorrefmark{1},
    \textbf{Shubhransh~Singhvi}\IEEEauthorrefmark{1}, \textbf{K.V.~Rashmi}\thanks{\hspace{-0.275cm}\IEEEauthorrefmark{1} Indicates equal contribution.\\
    This work was supported in part by the NSF CAREER Award under Grant 19434090 and in part by a Sloan Faculty Fellowship.}\\}
    \IEEEauthorblockA{
    Computer Science Department, Carnegie Mellon University, Pittsburgh, PA, USA}\qquad
    \IEEEauthorblockA{
    \\Email: saranshc@cs.cmu.edu, shubhranshsinghvi2001@gmail.com, \ rvinayak@cs.cmu.edu}
}

\title{Bandwidth Cost of Locally Repairable Convertible Codes in the Global Merge Regime}
\date{\today}
 \maketitle

\begin{abstract}
Recent studies have shown that distributed storage systems can achieve significant space savings by adapting redundancy levels to varying disk failure rates. This adaptation is performed via \emph{code conversion}, wherein data encoded under an initial code are transformed to data encoded under a final code. While this process is typically resource-intensive, \emph{convertible codes} are designed to enable these transformations efficiently while preserving desirable decodability constraints such as \emph{repair degree}, or the number of nodes accessed during node repair. 

In this work, we focus on the \emph{bandwidth cost} of conversion, or the total amount of data transferred during the conversion process. We study fundamental limits on the bandwidth cost of conversion between systematic optimal-distance Locally Repairable Codes (LRCs). We restrict our focus to the \emph{global merge regime}, in which multiple initial codewords are combined to form a single final codeword while preserving information locality. We focus on \emph{stable} convertible codes, wherein the number of \emph{unchanged nodes} is maximized during conversion. We generalize an information-theoretic approach for modeling code conversion to the LRC setting, and derive the first non-trivial lower bounds on the bandwidth cost of conversion in this regime. Notably, our bounds do not rely on any linearity assumptions. Consequently, we show that the constructions of Maturana and Rashmi are bandwidth-optimal across a broad range of parameters in the global merge regime.
\end{abstract}

\section{Introduction}

\emph{Erasure codes} are often employed in modern large-scale distributed storage systems (DSSs) to ensure data reliability and high availability while substantially reducing storage overhead compared to replication-based solutions~\cite{patterson1988case, ghemawat2003google, huang2012erasure}. An $(n,k)$ erasure code maps $k$ data symbols to a \emph{codeword} of length $n$, with the codeword symbols stored across distinct storage nodes. In large-scale DSSs, data are typically stored as multiple independently encoded codewords distributed over different subsets of storage devices.  Among all $(n,k)$ erasure codes, \emph{maximum distance separable} (MDS) codes achieve the optimal trade-off between storage overhead and failure tolerance by tolerating any $n-k$ erasures, as dictated by the Singleton bound~\cite{singleton1964maximum}. However, MDS codes incur high \emph{repair costs}, as recovering from any single node failure (erasure) in an MDS-coded DSS requires accessing $k$ other nodes. As repair operations occur frequently in practice, these costs can significantly degrade performance in large-scale DSSs~\cite{rashmi2013facebook,sathiamoorthy2013xoring,rashmi2014hitchhiker,huang2012erasure}. Moreover, there has recently been a growing interest in \emph{wide codes}\cite{kadekodi2020pacemaker,kadekodi2022tiger,kadekodi2023widelrcs}, i.e., codes with large dimension $k$, as they can achieve lower storage overhead for a prescribed level of reliability, further exacerbating repair costs of MDS-coded DSSs. \emph{Locally repairable codes} (LRCs) \cite{gopalan2012locality} mitigate this issue by enabling node repair through access to only a small number of other nodes at the expense of additional storage overhead. A substantial body of work has explored the fundamental limits, constructions, and trade-offs associated with LRCs; we refer the reader to the survey in~\cite{balaji2018erasurecodingoverview} for a comprehensive overview. In this work, we consider a class of systematic LRCs  parameterized by $(k,g,r,\ell)$, wherein each codeword consists of $k$ \emph{information symbols} partitioned into $\frac{k}{r}$ disjoint \emph{local groups} of $r$ information symbols and $\ell$ \emph{local parity symbols} each, along with $g$ \emph{global parity symbols} \cite{maturana2023LRCC}. A schematic illustration of this structure is shown in Figure~\ref{fig:lrc_structure}. In particular, we consider optimal-distance $(k,g,r,\ell)$-LRCs. 

\begin{figure}[t]
\centering
\begin{tikzpicture}[scale=0.9, every node/.style={font=\small}]

\def\r{3}
\def\ell{2}
\def\g{3}
\def\groups{3}

\def\dx{0.85}
\def\gap{.25}

\pgfmathsetmacro{\totalwidth}{\groups*(\r*\dx + \ell*\dx + \gap)}

\tikzstyle{info}=[draw, rectangle, minimum width=0.8cm, minimum height=0.8cm, fill=blue!20]
\tikzstyle{localpar}=[draw, rectangle, minimum width=0.8cm, minimum height=0.8cm, fill=green!30]
\tikzstyle{globalpar}=[draw, rectangle, minimum width=0.8cm, minimum height=0.8cm, fill=red!30]

\foreach \j in {1,...,3} {

    \pgfmathsetmacro{\base}{(\j-1)*(\r*\dx + \ell*\dx + \gap)}

    \foreach \i in {1,...,3} {
        \pgfmathtruncatemacro{\t}{(\j-1)*\r + \i}
        \node[info] at ({\base + (\i-1)*\dx},0) {};
    }

    \foreach \i in {1,...,2} {
        \pgfmathtruncatemacro{\p}{(\j-1)*\ell + \i}
        \node[localpar] at ({\base + (\r+\i-1)*\dx},0) {$\mathrm{L}$};
    }

    \node at ({\base + 2*\dx}, -1) {Local group};
}

\foreach \i in {1,...,3} {
    \node[globalpar] at ({\totalwidth + (\i-1)*\dx}, 0) {$\mathrm{G}$};
}
\node at ({\totalwidth + \dx}, -1) {Global Parities};

\end{tikzpicture}
\caption{A codeword of a $(k = 9, g = 3, r = 3, \ell = 2)$-LRC. Empty boxes denote information symbols. $\mathrm{L}$ and $\mathrm{G}$ denote local parity symbols and global parity symbols, respectively. }
\label{fig:lrc_structure}
\end{figure}
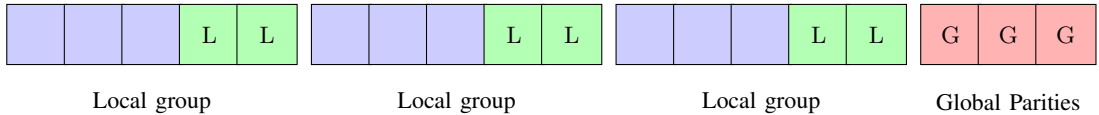

Recent studies\cite{kadekodi2019cluster, kim2024morph} highlight the importance of adapting the parameters of already encoded data to evolving storage device failure rates and changing access patterns. However, updating code parameters via the \emph{default approach} of fully re-encoding the data is prohibitively expensive~\cite{maturana2022convertible}. To address this issue, Maturana and Rashmi~\cite{maturana2022convertible} introduced the \emph{code conversion problem}, which provides a formal framework for efficiently updating the parameters of encoded data. In this framework, multiple codewords of an initial code $\mathcal{C}^I$ are transformed into multiple codewords of a final code $\mathcal{C}^F$, while minimizing the cost of conversion. The central goal is to jointly design the initial and final codes, along with a conversion procedure, such that this transformation can be performed more efficiently than naive re-encoding, subject to decodability constraints on $\mathcal{C}^I$ and $\mathcal{C}^F$, such as repair locality. A pair of codes that admits such a procedure is referred to as a \emph{convertible code}. 

Prior work on convertible codes has primarily focused on two cost metrics of conversion: the \emph{access cost}, defined as the total number of nodes read and written {during the conversion procedure}, and the \emph{bandwidth cost}, defined as the sum of \emph{read bandwidth cost} and \emph{write bandwidth cost}, the total amount of data read from and written to the nodes, respectively, {during the conversion procedure}. The problem of efficiently converting between LRCs was first considered by Maturana and Rashmi~\cite{maturana2023LRCC}. 

Despite the growing importance of LRCs in practical storage systems, the problem of efficiently converting between LRCs with different parameters remain largely unexplored \cite{maturana2023LRCC, Kong2023LocallyRC, ge2026locally, shi2026bounds}. In this work, we investigate the fundamental limits of the bandwidth cost of conversion between optimal-distance $(k,g,r,\ell)$-LRCs. In particular, we focus on the \emph{global merge regime}, in which multiple initial codewords are merged into a single final codeword, the number of global parity symbols is allowed to be updated arbitrarily, and all remaining code parameters are preserved; that is, conversion from an initial optimal-distance $(\Ik,\Ig,r,\ell)$-LRC to a final optimal-distance $(\Il\Ik,\Fg,r,\ell)$-LRC. A schematic illustration of this conversion is shown in Figure~\ref{fig:lrc_global_merge}.

We restrict our focus to \emph{stable} convertible codes, which maximize the number of symbols that remain unchanged during the conversion procedure. For this regime, Maturana and Rashmi proposed constructions~\cite{maturana2023LRCC} for which both the initial and final codes as well as the conversion procedure are linear. 

\begin{figure}[t]
\centering
\resizebox{0.75\linewidth}{!}{
\begin{tikzpicture}[scale=0.9, every node/.style={font=\small}]

\def\r{3}
\def\ell{1}
\def\g{3}
\def\groups{3}

\def\dx{0.85}
\def\gap{.25}
\def\yshift{-2.5}

\pgfmathsetmacro{\totalwidth}{\groups*(\r*\dx + \ell*\dx + \gap)}

\tikzstyle{info}=[draw, rectangle, minimum width=0.8cm, minimum height=0.8cm, fill=blue!20]
\tikzstyle{localpar}=[draw, rectangle, minimum width=0.8cm, minimum height=0.8cm, fill=green!30]
\tikzstyle{globalpar}=[draw, rectangle, minimum width=0.8cm, minimum height=0.8cm, fill=red!30]

\node[align=right] at (-2.5, 0) {\textbf{Initial codeword-1}};
\node[align=right] at (-2.5, \yshift) {\textbf{Initial codeword-2}};

\foreach \j in {1,...,3} {

    \pgfmathsetmacro{\base}{(\j-1)*(\r*\dx + \ell*\dx + \gap)}

    \foreach \i in {1,...,3} {
        \pgfmathtruncatemacro{\t}{(\j-1)*\r + \i}
        \node[info] at ({\base + (\i-1)*\dx},0) {};
    }

    \foreach \i in {1} {
        \pgfmathtruncatemacro{\p}{(\j-1)*\ell + \i}
        \node[localpar] at ({\base + (\r+\i-1)*\dx},0) {$\mathrm{L}$};
    }

}

\foreach \i in {1,...,3} {
    \node[globalpar] at ({\totalwidth + (\i-1)*\dx}, 0) {$\mathrm{G}$};
}

\foreach \j in {4,...,6} {

    \pgfmathsetmacro{\base}{(\j-4)*(\r*\dx + \ell*\dx + \gap)}

    \foreach \i in {1,...,3} {
        \pgfmathtruncatemacro{\t}{9 + (\j-4)*\r + \i}
        \node[info] at ({\base + (\i-1)*\dx}, \yshift) {};
    }

    \foreach \i in {1} {
        \pgfmathtruncatemacro{\p}{\ell*\groups + (\j-4)*\ell + \i}
        \node[localpar] at ({\base + (\r+\i-1)*\dx}, \yshift) {$\mathrm{L}$};
    }

}

\foreach \i in {4,...,6} {
    \node[globalpar] at ({\totalwidth + (\i-4)*\dx}, \yshift) {$\mathrm{G}$};
}

\node[align=right] at (-2, 3*\yshift - 0.55 ) {\textbf{Final codeword}};

\foreach \j in {1,...,3} {

    \pgfmathsetmacro{\base}{(\j-1)*(\r*\dx + \ell*\dx + \gap)}

    \foreach \i in {1,...,3} {
        \pgfmathtruncatemacro{\t}{(\j-1)*\r + \i}
        \node[info] at ({\base + (\i-1)*\dx},3*\yshift) {};
    }

    \foreach \i in {1} {
        \pgfmathtruncatemacro{\p}{(\j-1)*\ell + \i}
        \node[localpar] at ({\base + (\r+\i-1)*\dx},3*\yshift) {$\mathrm{L}$};
    }
}

\foreach \j in {4,...,6} {

    \pgfmathsetmacro{\base}{(\j-4)*(\r*\dx + \ell*\dx + \gap)}

    \foreach \i in {1,...,3} {
        \pgfmathtruncatemacro{\t}{9 + (\j-4)*\r + \i}
        \node[info] at ({\base + (\i-1)*\dx}, 3*\yshift - 1.1) {};
    }

    \foreach \i in {1} {
        \pgfmathtruncatemacro{\p}{\ell*\groups + (\j-4)*\ell + \i}
        \node[localpar] at ({\base + (\r+\i-1)*\dx}, 3*\yshift - 1.1) {$\mathrm{L}$};
    }

}

\foreach \i in {1,...,4} {
    \node[globalpar] at ({\totalwidth + (\i-1)*\dx}, 3*\yshift - 0.55) {$\mathrm{G}$};
}

\draw[thick, dashed, rounded corners]
(-4.2, 0.8) rectangle ({\totalwidth + 3*\dx + 0.5}, \yshift - 1.3);

\draw[thick, dashed, rounded corners]
(-3.5, 3*\yshift + 0.8) rectangle ({\totalwidth + 3.5*\dx + 0.2}, 3*\yshift - 1.8);

\draw[->, thick]
({\totalwidth/2}, 1.6*\yshift )
--
({\totalwidth/2}, 2.6*\yshift)
node[midway, right]{};

\end{tikzpicture}
}

\caption{Conversion from a $(\Ik=9,\Ig=3,r=3,\ell=1)$-LRC to a $(\Fk=18,\Fg=4,r=3,\ell=1)$-LRC. Empty boxes denote information symbols. $\mathrm{L}$ and $\mathrm{G}$ denote local parity symbols and global parity symbols, respectively.}
\label{fig:lrc_global_merge}
\end{figure}
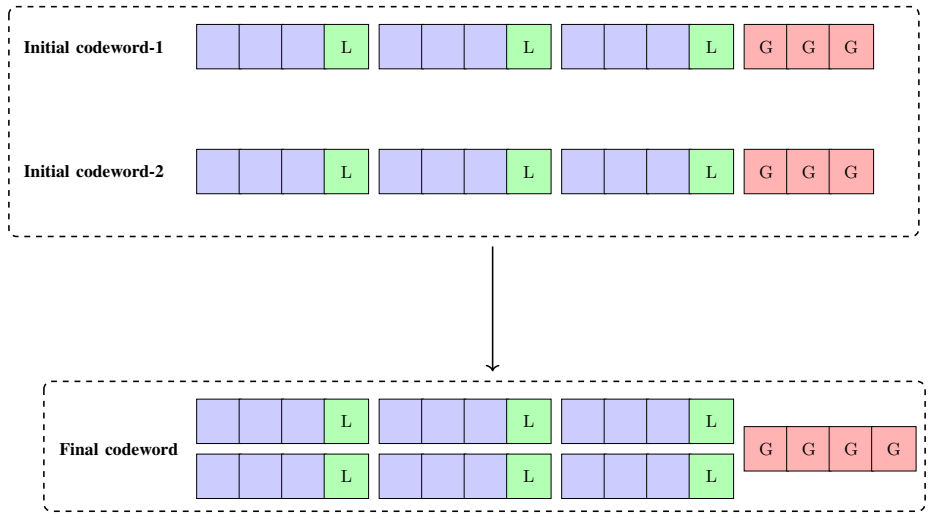

 In this work, we derive lower bounds on the bandwidth cost of conversion between optimal-distance $(k,g,r,\ell)$-LRCs in the global merge regime without imposing any linearity assumptions. We do so by using an information-theoretic approach which generalizes the approach introduced in~\cite{singhvi2025tightlowerboundsbandwidth} to the setting of {optimal-distance} LRCs. We first derive an information-theoretic constraint on the data downloaded during code conversion (Lemma~\ref{lem:bandwidth_constraint}). We then use this constraint to obtain a lower bound on the read bandwidth cost in Theorem~\ref{thm:optimalLRC}. In doing so, we show that the constructions of \cite{maturana2023LRCC} are bandwidth-optimal for the global merge regime when $\Fg \leq r$.

\textbf{Outline of the paper:} 
Section~\ref{sec:background} reviews the necessary background on locally repairable codes and convertible codes. In Section~\ref{sec:LRC_Storage}, we  establish fundamental information-theoretic inequalities for LRC storage that serve as essential tools for our analysis. In Section~\ref{sec:bandwidth_cost}, we derive lower bounds on the bandwidth cost of conversion between {optimal-distance} LRCs in the global merge regime and identify the parameter regimes in which these bounds are tight. Finally, Section~\ref{sec:conclusion} summarizes the main results.

\section{Background and Notation}\label{sec:background}

Convertible codes refer to a class of erasure codes that allow efficient tuning of code parameters while retaining desirable properties such as systematic encoding, MDS guarantees~\cite{chopra2024low, maturana2020access, maturana2022bandwidth, maturana2023bandwidth, ge2024mds, ramkumar2025mds, singhvi2025tightlowerboundsbandwidth, wang2025lowerboundsconversionbandwidth} or repair locality~\cite{maturana2023LRCC, Kong2023LocallyRC, ge2026locally, shi2026bounds}. We begin by introducing the relevant notation and preliminary definitions prior to the formal definition.

The notation in this paper builds on the notation introduced in~\cite{singhvi2025tightlowerboundsbandwidth}. For integers $a, i \in \mathbb{Z}_{\geq 1}$, let $[a]^i$ denote the set of integers $[a]^i\;=\;\{a(i-1)+1, a(i-1)+2, \ldots, ai\}$ and let $[a]\,:=\,[a]^{1}$. For a vector $\bfx$ of length $n$ and a subset $\cA \subseteq [n]$, the notation $\bfx_{\cA}$ represents the projection of $\bfx$ onto the indices specified by the set $\cA$. For a prime power $q$, let $\mathbb{F}_q$ denote the finite field of size $q$. For any set $\cA$, and any collection of random variables $\{X_i\}_{i \in \cA}$, let $X_\cA$ denote $\{X_i\}_{i \in \cA}$. Furthermore, for any functions $\{f_i\}_{i \in \cA}$, let $f_\cA(X_\cA)$ denote $\{f_i(X_i)\}_{i \in \cA}$.

An $(n,k)$ code $\mathcal{C}$ over a finite field $\mathbb{F}_q$ is defined by an injective  map $\mathcal{C}: \mathbb{F}_q^k \to \mathbb{F}_q^n$. A code $\mathcal{C}$ is \emph{systematic} if it maps a message vector $\bfm \in \mathbb{F}_q^{k}$ to a codeword that includes all the symbols of $\bfm$, termed \emph{information symbols}, uncoded. The indices of the codewords storing the information symbols are termed \emph{information coordinates}. For any code $\cC \subseteq\Fqn$ and subset $\cS \subseteq [n]$, let the \emph{punctured code} $\cC_\cS$ denote the projection of $\cC$ onto the coordinates in $\cS$. 

For any random variables $X$ and $Y$, we denote the entropy of $X$ by $\Hent{X}$ and the conditional entropy of $X$ given $Y$ by $\condH{X}{Y}$. The mutual information between $X$ and $Y$ is denoted by $\Iinfo{X}{Y}$.

\subsection{Locally Repairable Codes}
We now present the commonly-accepted definition of ($r,\delta$)-locality  for codes that are not assumed to be linear.

\begin{definition}[($r,\delta$)-locality \cite{prakash2012linearLRC, westerback2016combinatoricslrc, RawatLRC14}]\label{def:locality}
Let $\cC \subseteq \Fqn$ be a code of length $n$, and let $1 \le r \le k$ and $\delta \ge 2$. A subset $S \subseteq [n]$ is termed an \emph{$(r,\delta)$-locality set} of $\cC$ if

\begin{enumerate}[(i)]
    \item $|S| \le r + \delta - 1$, 
    \item the minimum distance of the punctured code $\cC_S$ satisfies $d(\cC_S) \ge \delta$.
\end{enumerate}  Then, $\cC$ has \emph{$(r,\delta)$-all-symbol locality} if every coordinate
$t \in [n]$ belongs to at least one $(r,\delta)$-locality set.  

Furthermore, for a systematic code $\cC \subseteq \Fqn$, let $[k]$ be the information coordinates without loss of generality. Then, $\cC$ has \emph{$(r,\delta)$-information locality} if every coordinate
$t \in [k]$ belongs to at least one $(r,\delta)$-locality set.  
\end{definition}

Next, we introduce the class of locally repairable codes we will be focusing on in this work.

\begin{definition}[$(k,g,r,\ell)$-LRC]\label{def:kgrlLRC}
Let $\cC$ be an {$(n,k)$ systematic} code. Let  $r \mid k$ and $ n \;:=\; k + \frac{k\ell}{r} + g.$ The code $\cC$ is termed a \emph{$(k,g,r,\ell)$-LRC} if it satisfies the following properties:
\begin{enumerate}[(i)]
    \item $\cC$ has $(r,\ell+1)$-information locality. 
    \item The $k$ information symbols are partitioned into
    $\mu := \frac{k}{r}$ disjoint locality sets, also termed \emph{local groups}, each containing $r$ information symbols
    and $\ell$ associated \emph{local parity symbols} that depend only on the corresponding information symbols. 
    \item Each codeword also consists of $g$ \emph{global parity} symbols, each of which depends on some subset of the $k$ information symbols. 
\end{enumerate}
\end{definition}

Let $\cC$ be a $(k,g,r,\ell)$-LRC. Observe that if all the global parity symbols, all the local parity symbols of a local group, and an information symbol from the same local group are erased, then by Definition~\ref{def:kgrlLRC}, such an erasure pattern is not recoverable. Hence, the minimum distance $d$ of $\cC$ satisfies 
\begin{equation}\label{eq:lrc-distance-bound}
    d \;\le\; g + \ell + 1. 
\end{equation}
Furthermore, if the above bound is met with equality, then $\cC$ is termed an \emph{optimal-distance} $(k,g,r,\ell)$-LRC. 

In this work, we study the code conversion problem specifically for optimal-distance $(k,g,r,\ell)$-LRCs.

\subsection{Convertible Codes}
Next, we provide the formal framework for studying code conversions introduced in \cite{maturana2020access, maturana2022convertible}.

\begin{definition}[Convertible Code {\cite{maturana2020access, maturana2022convertible}}]  
A \emph{convertible code} with parameters $(\In,\Ik;\Fn,\Fk)$ consists of:
\begin{enumerate}[(i)]
    \item An \emph{initial}
    {$(\In,\Ik)$} code $\mathcal{C}^I$, a \emph{final} {$(\Fn,\Fk)$} code $\mathcal{C}^F$.
    \item A pair of partitions $\IPart$ and $\FPart$ of the common message index set 
          $[M]$ with $M = \mathrm{lcm}(\Ik,\Fk)$, such that each block in $\IPart$ has size $\Ik$ 
          and each block in $\FPart$ has size $\Fk$.
    \item A  \emph{conversion procedure} that, for all 
          $\bfm \in \mathbb{F}_q^M$, takes as input
          $\{\mathcal{C}^I(\bfm_P) : \cP \in \IPart\}$ and produces as output
          $\{\mathcal{C}^F(\bfm_P) : \cP \in \FPart\}$.
\end{enumerate}
\end{definition}

During the conversion procedure, storage nodes can be classified into three types: \emph{unchanged nodes}, which remain part of the final codewords without modification; \emph{retired nodes}, which are removed during the conversion; and \emph{new nodes}, which are introduced in the final codewords and were not present in the initial codewords. The \emph{conversion coordinator}, which manages the conversion process, downloads data from both unchanged and retired nodes to construct the new nodes. Convertible codes that maximize the number of unchanged nodes are termed \emph{stable codes}. For stable codes, it suffices to study the read bandwidth cost, since the write bandwidth cost is fixed and independent of the conversion procedure. In this work, we restrict our focus to stable convertible codes.

\subsection{Locally Repairable Convertible Codes}

In the following definition, we formalize the notion of locally repairable convertible codes (LRCCs).

\begin{definition}[$(\Ik,\Ig,\Ir,\ell^I;\Fk,\Fg,\Fr,\ell^F)$-LRCC]
        When $\cC^I$ and $\cC^F$ are $(\Ik,\Ig,\Ir,\ell^I)$- and $(\Fk,\Fg,\Fr,\ell^F)$-LRCs, respectively, the corresponding convertible code is termed a \emph{locally repairable convertible code} (LRCC) with parameters $(\Ik,\Ig,\Ir,\ell^I;\Fk,\Fg,\Fr,\ell^F)$. Moreover, when $\cC^I$ and $\cC^F$ are optimal-distance, the corresponding convertible code is termed an optimal-distance LRCC.
\end{definition}

In the remainder of this work, the term LRCC refers to an $(\Ik,\Ig,\Ir,\ell^I;\Fk,\Fg,\Fr,\ell^F)$-LRCC. For LRCCs, let $\mu^I := \frac{\Ik}{\Ir}$ and $\mu^F := \frac{\Fk}{\Fr}$ denote the number of local groups of the initial and final codewords, respectively.

\begin{definition}[Global Merge Regime]
Within the class of LRCCs, the \emph{global merge regime} corresponds to the subclass of conversions wherein multiple codewords are merged into a single codeword while maintaining locality parameters. That is, $\Fk = \Il\Ik$, $\Fr = \Ir$, and $\ell^F = \ell^I$. Moreover, $\mu^F = \Il\mu^I$.
\end{definition}

We make the following observation for LRCCs in the global merge regime, namely $(\Ik,\Ig,r,\ell;\Il\Ik,\Fg,r,\ell)$-LRCCs. 

\begin{remark}
\label{rem:sys_unchanged_LP}
For any $(\Ik,\Ig,r,\ell;\Il\Ik,\Fg,r,\ell)$-LRCC, both the initial and final codes are systematic and encode the same data. Consequently, the information nodes of the initial codewords remain unchanged during the conversion procedure and directly serve as the corresponding information nodes in the final codewords; hence, no rewriting of these nodes is required.
\end{remark}

In~\cite{maturana2023LRCC}, Maturana and Rashmi presented constructions of stable optimal-distance
$(\Ik,\Ig,r,\ell;\Fk,\Fg,r,\ell)$-LRCCs. Leveraging the
piggybacking framework of~\cite{rashmi2017piggybacking}, they designed convertible
codes that significantly reduce the conversion bandwidth compared to the default
re-encoding approach. Let $\bw$ denote the read bandwidth cost. For the global merge regime, the constructions in~\cite{maturana2023LRCC} achieve 
\begin{equation*}
\bw
=
\begin{cases}
\Il\Fg\alpha , & \text{if } \Fg \le \Ig, \\[0.6em]
\Il\left(
\dfrac{(\Ik + \Imm\ell)(\Fg - \Ig)}{\Fg + \ell}
+ g^{I}
\right)\alpha, & \text{otherwise},
\end{cases}
\end{equation*}
where $\alpha$ denotes per-node storage capacity over the finite field $\mathbb{F}_q$.   

In this work, we derive lower bounds on the bandwidth cost of stable optimal-distance LRCCs in the global merge regime and our bounds demonstrate that the LRCC constructions of \cite{maturana2023LRCC} are bandwidth-optimal for the parameter regime $\Fg \leq r$.

\subsection{Related Work}

The approach used in this paper is a generalization of that of~\cite{singhvi2025tightlowerboundsbandwidth}, which utilized similar tools to derive information-theoretic lower bounds on the bandwidth cost of MDS convertible codes in the \textit{split regime}. We adapt the approach to study code conversion for LRCs, thus extending its applicability beyond the MDS setting. While the methodology is similar, the parameter regime considered in this work is fundamentally different: we focus on the global merge regime. 

Information-theoretic techniques have been previously used to establish fundamental limits in distributed storage, including the non-achievability of interior points on the storage–repair–bandwidth tradeoff for regenerating codes~\cite{shah2011distributed}. 

The access cost of conversion between LRCs with $(r,2)$-all-symbol locality was first studied in \cite{Kong2023LocallyRC}. These results were generalized to LRCs with $(r,\delta)$-all-symbol locality in both~\cite{ge2026locally} and~\cite{shi2026bounds}. In \cite{Justin2025}, the authors designed access-optimal convertible codes with information-theoretic security in the presence of passive eavesdroppers. Recently, in \cite{gruica2026convertiblecodesdatadevice}, the authors investigated the access cost of general linear codes and Reed-Muller codes in the merge regime. A related problem is the \emph{scaling problem}~\cite{
zhang2010alv,
zheng2011fastscale,
wu2012gsr,
zhang2014rethinking,
huang2015scaleRS,
wu2016ioefficient,
zhang2018optimal,
hu2018generalized,
zhang2020efficient,
rai2015adaptive,
rai2015adaptive2,
wu2020optimal,
hu2021combinedlocality,
wu2022optimaltradeoff,
wu2022placement
}, which involves converting an $[n,k,\alpha]$ code into an
$[n+s,k+s,\tfrac{k\alpha}{k+s}]$ code while preserving the total amount of stored
data and redistributing it across a different number of storage nodes. Another related problem is the \emph{data rebalancing problem}~\cite{CDR2020, CDRDecentralized2020, CDRCyclic2022}, which studies the efficient redistribution of replicated data when storage nodes are added to or removed from a system.

\section{Information-Theoretic Inequalities for LRC Storage}\label{sec:LRC_Storage}
In this section, we develop information-theoretic inequalities that characterize the structure of LRC storage. Our approach generalizes the one presented in~\cite{singhvi2025tightlowerboundsbandwidth} for {studying the conversion problem for} MDS-coded storage by incorporating the structural constraints imposed by locality. This serves as the foundation for the lower bounds on the bandwidth cost of LRCCs presented in the subsequent section.

Consider a DSS encoded using an optimal-distance $(k, g, r, \ell)$-LRC. The system comprises several independent codewords, and in the following analysis, we focus on one such codeword. Such a codeword is distributed across $k + \mu\ell + g$ nodes, where $\mu := k/r$. These consist of $k$ \emph{information nodes}, each storing a message symbol directly, and $\mu\ell + g$ \emph{parity nodes}, storing deterministic functions of the message symbols.

\begin{figure}[t]
\centering
\begin{tikzpicture}[scale=0.9, every node/.style={font=\small}]

\def\r{3}
\def\ell{2}
\def\g{3}
\def\groups{3}

\def\dx{0.85}
\def\gap{.25}

\pgfmathsetmacro{\totalwidth}{\groups*(\r*\dx + \ell*\dx + \gap)}

\tikzstyle{info}=[draw, rectangle, minimum width=0.8cm, minimum height=0.8cm, fill=blue!20]
\tikzstyle{localpar}=[draw, rectangle, minimum width=0.8cm, minimum height=0.8cm, fill=green!30]
\tikzstyle{globalpar}=[draw, rectangle, minimum width=0.8cm, minimum height=0.8cm, fill=red!30]

\foreach \j in {1,...,3} {

    \pgfmathsetmacro{\base}{(\j-1)*(\r*\dx + \ell*\dx + \gap)}

    \foreach \i in {1,...,3} {
        \pgfmathtruncatemacro{\t}{(\j-1)*\r + \i}
        \node[info] at ({\base + (\i-1)*\dx},0) {$X_{\t}$};
    }

    \foreach \i in {1,...,2} {
        \pgfmathtruncatemacro{\p}{(\j-1)*\ell + \i}
        \node[localpar] at ({\base + (\r+\i-1)*\dx},0) {$L_{\p}$};
    }

    \node at ({\base + 2*\dx}, -1) {Local group-$\j$};
}

\foreach \i in {1,...,3} {
    \node[globalpar] at ({\totalwidth + (\i-1)*\dx}, 0) {$G_{\i}$};
}
\node at ({\totalwidth + \dx}, -1) {Global Parities};

\end{tikzpicture}

\caption{A codeword of a $(k = 9, g = 3, r = 3, \ell = 2)$-LRC. The notation for the random variables corresponding to the data stored in information and local parity nodes are globally indexed.}
\label{fig:lrc_structure_notations}
\end{figure}
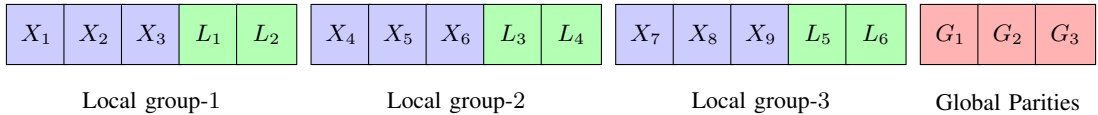

The $k$ information nodes are divided into $\mu$ disjoint local groups of size $r$. Each group is associated with $\ell$ \emph{local parity nodes} that store the local parity symbols; these symbols are functions exclusively of the $r$ information symbols in that specific group. Let $\tau \in [\mu]$ be the index of an arbitrary local group, and let $X_{[r]^\tau}$ and $L_{[\ell]^\tau}$ denote the random variables corresponding to the data stored in the $r$ information nodes and $\ell$ local parity nodes, respectively, comprising the $\tau$-th local group (see Figure~\ref{fig:lrc_structure_notations}). It holds that 
\begin{equation}
    \condH{L_{[\ell]^\tau}}{X_{[r]^\tau}} = 0\;.
\end{equation}
The remaining $g$ nodes are \emph{global parity nodes} that store the global parity symbols, each of which is a deterministic function of the $k$ information symbols. Let $j \in [g]$ be the index of an arbitrary global parity node, and  let $G_{j}$ denote the random variable corresponding to the data stored in the $j$-th global parity node. It holds that 
\begin{equation}
    \condH{G_{[g]}}{X_{[k]}} = 0\;.
\end{equation}

Let $\cQ \subseteq[g]$, $\cR \subseteq[\mu\ell]$, and $\cS \subseteq[k]$ be arbitrary subsets of indices of global parity nodes, local parity nodes, and information nodes, respectively, such that $\lvert \cQ\rvert + \lvert \cR\rvert + \lvert \cS\rvert\;\leq\;g + \ell$. Since the DSS is encoded using an optimal-distance $(k, g, r, \ell)$-LRC, it holds that
\begin{equation}\label{eq:inf_dist_prop}
    \condH{G_{\cQ},L_{\cR},X_{\cS}}{G_{[g]\setminus\cQ},L_{[\mu\ell]\setminus\cR},X_{[k]\setminus\cS}} = 0\;.
\end{equation}

We assume that each node has a storage capacity of $\alpha$ symbols over $\mathbb{F}_q$, and therefore cannot store data with entropy exceeding $\alpha$. Thus, let $j \in [k]$, $i \in [g]$, and $a \in [\mu\ell]$ be the indices of arbitrary information nodes, global parity nodes, and local parity nodes, respectively, and it follows that
\begin{equation}
    \Hent{X_{j}} \leq \alpha,\quad \Hent{G_{i}} \leq \alpha,\quad\Hent{L_{a}} \leq \alpha\;.
\end{equation}

Finally, we assume the random variables $X_{[k]}$ are independent and uniformly distributed, implying $\Hent{X_{j}} = \alpha$.

\begin{proposition}\label{prop:joint_entropy}
The total entropy of the data stored across the $k$ information nodes, the $\mu\ell$ local parity nodes, and the $g$ global parity nodes is \[
\Hent{X_{[k]}, L_{[\mu\ell]}, G_{[g]}}\;=\;k\alpha\;.\] 
\end{proposition}

Next, we introduce a series of information-theoretic results. We begin by demonstrating that for an optimal-distance $(k,g,r,\ell)$-LRC, any $r$ nodes from a set of $r+\ell+g$ nodes corresponding to any single local group and the global parity nodes are conditionally independent with respect to the remaining local groups and uniformly distributed. That is, the $r+\ell+g$ nodes, conditioned on the remaining local groups, form an $(r+\ell+g,r)$ MDS code.

\begin{proposition}\label{prop:independence} 
Let $\tau \in [\mu]$ be the index of an arbitrary local group. Let $\cA \subseteq [g]$, $\cB \subseteq [\ell]^\tau$, and $\cD \subseteq [r]^\tau$ be any subsets of indicies of global parity nodes, local parity nodes of the $\tau$-th local group, and information nodes of the $\tau$-th local group, respectively, such that $|\cA| + |\cB| + |\cD| \le r$. Then, $G_{\cA}, L_{\cB}, X_{\cD}$ are conditionally independent with respect to $X_{[k]\setminus[r]^\tau}$ and uniformly distributed. That is,
\[
\condH{G_{\cA},L_{\cB},X_{\cD}}{X_{[k]\setminus[r]^\tau}} = (|\cA| + |\cB| + |\cD|)\alpha \;. 
\]
\end{proposition}
\begin{IEEEproof}
    Observe that it suffices to show this result for all sets $\cA \subseteq [g], \cB \subseteq [\ell]^\tau$, and $\cD \subseteq [r]^\tau$ such that $|\cA| + |\cB| + |\cD| = r$. It follows that, as desired,
    \begin{align*}
        \condH{G_{\cA},L_{\cB},X_{\cD}}{X_{[k]\setminus[r]^\tau}}\;&=\;\Hent{X_{[k]\setminus[r]^\tau},G_{\cA},L_{\cB},X_{\cD}} - \Hent{X_{[k]\setminus[r]^\tau}}\\
        \;&=\;\Hent{X_{[k]\setminus[r]^\tau},L_{[\mu\ell]\setminus[\ell]^\tau},G_{\cA},L_{\cB},X_{\cD}} - \Hent{X_{[k]\setminus[r]^\tau}}\\
        \;&=\;\Hent{X_{[k]}, L_{[\mu\ell]}, G_{[g]}} -  \Hent{X_{[k]\setminus[r]^\tau}}\\
        &\qquad  -\; \condH{G_{[g]\setminus\cA},L_{[\ell]^\tau\setminus\cB},X_{[r]^\tau\setminus\cD}}{X_{[k]\setminus[r]^\tau},L_{[\mu\ell]\setminus[\ell]^\tau},G_{\cA},L_{\cB},X_{\cD}}\\
        \;&=\;\Hent{X_{[k]}, L_{[\mu\ell]}, G_{[g]}} - \Hent{X_{[k]\setminus[r]^\tau}}\tag{{Equation~\eqref{eq:inf_dist_prop}}: optimal-distance}\\
        \;&=\;k\alpha - (k-r)\alpha\tag{Proposition~\ref{prop:joint_entropy}, Independence of $X_{[k]}$}\\
        \;&=\;(|\cA| + |\cB| + |\cD|)\alpha\;.
    \end{align*}
\end{IEEEproof}

Next, we establish the structural property of optimal-distance $(k,g,r,\ell)$-LRCs that every global parity symbol must be a function of all $k$ information symbols. This will be used to derive properties of parity nodes of optimal-distance LRCCs in the following section.

\begin{proposition}\label{proposition:LRC_structural_property}
    For any optimal-distance $(k,g,r,\ell)$-LRC, each of the $g$ global parity nodes is a function of all $k$ information nodes.
\end{proposition}
\begin{IEEEproof}
    Assume to the contrary that there exists a global parity node that is not a function of all $k$ information nodes. Let $i \in [g]$ and $j \in [k]$ be arbitrary indices of global parity nodes and information nodes, respectively, such that the $i$-th global parity node is only a function of the information nodes indexed by $[k]\setminus\{j\}$. Consider the erasure pattern consisting of all global parity nodes except the $i$-th one, the $j$-th information node, and all of the local parity nodes of its local group. Such a pattern is not recoverable as by Definition~\ref{def:kgrlLRC}, none of the remaining nodes are functions of the $j$-th information node. However, this pattern is of size $g + \ell$, contradicting the assumption of optimal-distance. Therefore, each of the $g$ global parity nodes is a function of all $k$ information nodes, as desired.
\end{IEEEproof}

The following lemma establishes a general upper bound on the conditional mutual information between functions of random variables, which will be useful in our later analysis.

\begin{lemma}\label{lem:conditional_mutual_information_bound}
    Let $Z_1, Z_2, \dots,Z_n$ be a collection of random variables. Given any deterministic functions $\{f_i\}_{i \in [n]}$, disjoint subsets $\cA, \cB, \cD \subseteq [n]$, and subsets $\cA' \subseteq \cA,\cB' \subseteq \cB$ such that $Z_{(\cA \cup \cB)\setminus(\cA' \cup \cB')}$ are conditionally independent with respect to $Z_{\cD}$,
    \begin{align*}
        &\condI{f_\cA(Z_\cA)}{f_\cB(Z_\cB)}{f_{\cD}(Z_{\cD})}\;\leq\;\condH{f_{\cA'}(Z_{\cA'})}{f_{\cD}(Z_{\cD})}\;+\;\condH{f_{\cB'}(Z_{\cB'})}{f_{\cD}(Z_{\cD})}\;.
    \end{align*}
\end{lemma}

\begin{IEEEproof}
Observe that, as desired,
\begin{align*}
    &\condI{f_\cA(Z_\cA)}{f_\cB(Z_\cB)}{f_{\cD}(Z_{\cD})} \\
    \;=\;&\condH{f_\cA(Z_\cA)}{f_{\cD}(Z_{\cD})} + \condH{f_\cB(Z_\cB)}{f_{\cD}(Z_{\cD})}  -\condH{f_{\cA \cup \cB}(Z_{\cA \cup \cB})}{f_{\cD}(Z_{\cD})}\\
    \;\leq\;&\condH{f_{\cA'}(Z_{\cA'})}{f_{\cD}(Z_{\cD})} + \condH{f_{\cB'}(Z_{\cB'})}{f_{\cD}(Z_{\cD})} 
      +\; \condH{f_{\cA\setminus\cA'}(Z_{\cA\setminus\cA'})}{f_{\cD}(Z_{\cD})}  \\
    &\quad+ \condH{f_{\cB\setminus\cB'}(Z_{\cB\setminus\cB'})}{f_{\cD}(Z_{\cD})}  -\;\condH{f_{\cA \cup \cB}(Z_{\cA \cup \cB})}{f_{\cD}(Z_{\cD})} \tag{Sub-additivity}\\
    \;=\;&\condH{f_{\cA'}(Z_{\cA'})}{f_{\cD}(Z_{\cD})} + \condH{f_{\cB'}(Z_{\cB'})}{f_{\cD}(Z_{\cD})} 
     +\; \condH{f_{(\cA \cup \cB)\setminus(\cA' \cup \cB')}(Z_{(\cA \cup \cB)\setminus(\cA' \cup \cB')})}{f_{\cD}(Z_{\cD})} \\
    &\quad -\; \condH{f_{\cA \cup \cB}(Z_{\cA \cup \cB})}{f_{\cD}(Z_{\cD})}\tag{Independence}\\
    \;\leq\;&\condH{f_{\cA'}(Z_{\cA'})}{f_{\cD}(Z_{\cD})} + \condH{f_{\cB'}(Z_{\cB'})}{f_{\cD}(Z_{\cD})}\;.
\end{align*}
\end{IEEEproof}

Next, we state the following result from \cite{singhvi2025tightlowerboundsbandwidth}, which will be used in conjunction with the previous lemma to derive a lower bound on the read bandwidth cost of stable optimal-distance LRCCs in the following section. 

\begin{lemma}\cite[Lemma 2]{singhvi2025tightlowerboundsbandwidth}\label{lem:min_avg_entropy_bound}
    Let $Z_1,\dots,Z_b$ be random variables and let $\{f_i\}_{i \in [b]}$ be deterministic functions. For any $a \le b$, where $\cA_a$ is the set of all subsets of $[b]$ of size $a$, $$\underset{\cA \in \cA_a}{\min} \;\Hent{f_\cA(Z_\cA)}\;\le\; \frac{a}{b}\,\sum_{i \in [b]}\Hent{f_{i}(Z_{i})}\;.$$
    If $Z_1,\dots,Z_b$ are independent, then $$\underset{\cA \in \cA_a}{\min} \;\Hent{f_\cA(Z_\cA)}\;\le\; \frac{a}{b}\,\Hent{f_{[b]}(Z_{[b]})}\;.$$
\end{lemma}

\section{Bandwidth Cost of LRCCs in the Global Merge Regime}\label{sec:bandwidth_cost}

In this section, we generalize the information-theoretic approach for studying code conversion developed in~\cite{singhvi2025tightlowerboundsbandwidth} for the setting of LRCCs to establish lower bounds on the bandwidth cost of stable optimal-distance $(\Ik, \Ig, r, \ell; \Fk = \Il \Ik, \Fg, r, \ell)$-LRCCs. Recall that in the global merge regime, multiple initial codewords are merged into a single final codeword. As discussed in Remark~\ref{rem:sys_unchanged_LP}, for LRCCs, the initial information nodes remain unchanged during the global merge conversion procedure, and serve as the information nodes in the final codeword.

The following lemma is a direct consequence of the optimal-distance of the initial and final LRCs and the independence of the initial codewords and local groups in the global merge regime.
\begin{lemma}
For any optimal-distance $(\Ik, \Ig, r, \ell; \Il\Ik, \Fg, r, \ell)$-LRCC, the following hold: 
\begin{enumerate}[(i)]
    \item All final global parity nodes are new nodes.
    \item All initial global parity nodes are retired nodes.
    \item All unchanged initial local parity nodes must serve as local parity nodes for the same set of information nodes in the final codeword.
\end{enumerate}
\end{lemma}
\begin{IEEEproof} We first prove that all final global parity nodes are newly introduced. Assume to the contrary that there exists a final global parity node that is unchanged. It must serve as a parity node for some initial codeword and thus is a function of at most $\Ik$ information nodes. As $\Il \ge 2$, this contradicts the structural property of optimal-distance LRCs of the final codeword that every global parity node must be a function of all $\Il\Ik$ information nodes (Proposition~\ref{proposition:LRC_structural_property}). Thus, all final global parity nodes are new nodes.

Next, we similarly prove that all initial global parity nodes are retired. We first separately consider the case that $\Ik = r$. In this case, the initial codewords are all $(r+\ell+\Ig,r)$ MDS codes, and thus initial local parity nodes and initial global parity nodes are equivalent up to relabeling. We claim that we can label all unchanged initial parity nodes as local parity nodes so that all initial global parity nodes must be retired. Assume to the contrary that there exists an initial codeword with at least $\ell + 1$ unchanged parity nodes. By the first part of this proof, the unchanged parity nodes must all serve as local parity nodes in the final codeword. As they are all functions of the same set of $r$ information symbols, and final local groups are disjoint, the unchanged parities must all serve as local parities for the same final local group. This is a contradiction as each final local group contains only $\ell$ local parities. Thus, without loss of generality, when $\Ik = r$, all initial global parity nodes are retired.

Next, let us consider the case that $\Ik > r$. Assume to the contrary that there exists an initial global parity node that remains unchanged. Let $t \in [\Il]$ be the index of an arbitrary initial codeword with an unchanged global parity node. By the first part of this proof, the unchanged initial global parity node must serve as a local parity node in some local group of the final codeword and is thus a function of $r$ information nodes. As $\Ik > r$, this contradicts the structural property of optimal-distance LRCs of the $t$-th initial codeword that every global parity node must be a function of all $\Ik$ information nodes (Proposition~\ref{proposition:LRC_structural_property}). Thus, in all cases, all initial global parity nodes are retired.

Finally, note that by the first part of this proof, any unchanged initial local parity node must serve as a local parity node in the final code as well. Moreover, by Definition~\ref{def:locality}, each local parity is a function of each one of the $r$ information nodes in its local group, so an unchanged local parity node must serve as a local parity node for the same set of information nodes in the final codeword as it did in the initial codeword.
\end{IEEEproof}
Consequently, for stable optimal-distance LRCCs in the global merge regime, the only new nodes are the final global parity nodes, and the local groups, including the local parity nodes themselves, remain unchanged from the initial codewords to the final codeword. For the remainder of this paper, we focus only on stable optimal-distance LRCCs in the global merge regime.

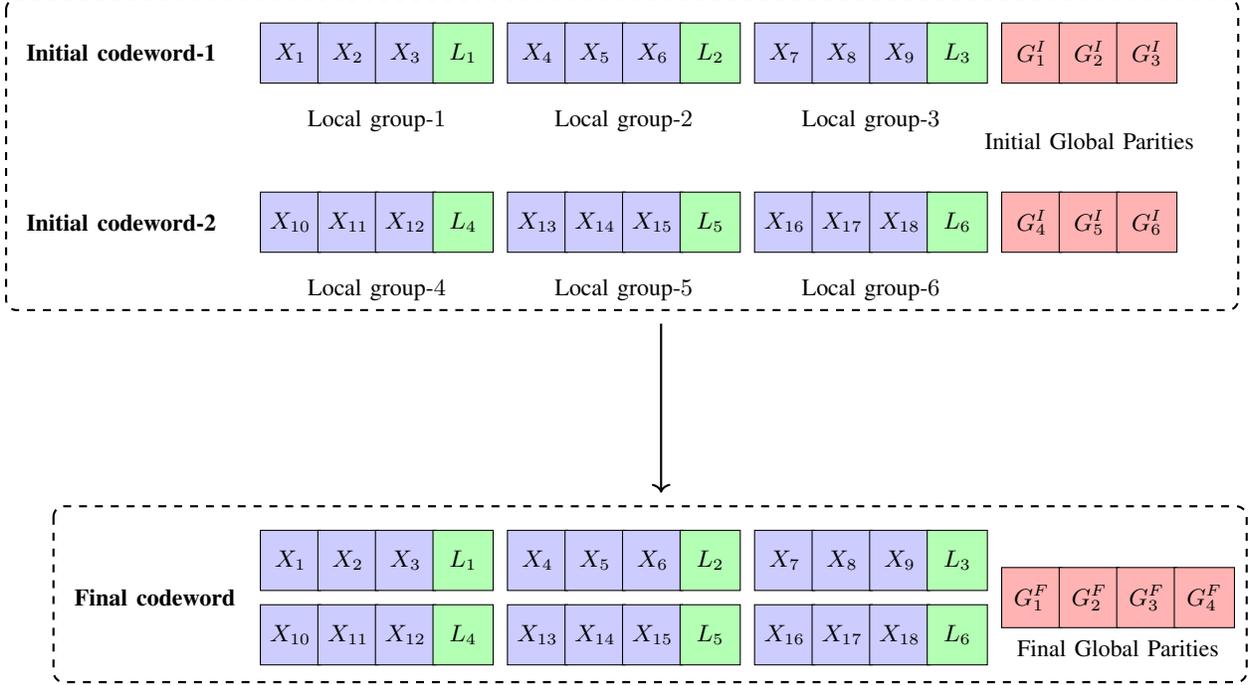
\begin{figure}[t]
\centering
\begin{tikzpicture}[scale=0.9, every node/.style={font=\small}]

\def\r{3}
\def\ell{1}
\def\g{3}
\def\groups{3}

\def\dx{0.85}
\def\gap{.25}
\def\yshift{-2.5}

\pgfmathsetmacro{\totalwidth}{\groups*(\r*\dx + \ell*\dx + \gap)}

\tikzstyle{info}=[draw, rectangle, minimum width=0.8cm, minimum height=0.8cm, fill=blue!20]
\tikzstyle{localpar}=[draw, rectangle, minimum width=0.8cm, minimum height=0.8cm, fill=green!30]
\tikzstyle{globalpar}=[draw, rectangle, minimum width=0.8cm, minimum height=0.8cm, fill=red!30]

\node[align=right] at (-2.5, 0) {\textbf{Initial codeword-1}};
\node[align=right] at (-2.5, \yshift) {\textbf{Initial codeword-2}};

\foreach \j in {1,...,3} {

    \pgfmathsetmacro{\base}{(\j-1)*(\r*\dx + \ell*\dx + \gap)}

    \foreach \i in {1,...,3} {
        \pgfmathtruncatemacro{\t}{(\j-1)*\r + \i}
        \node[info] at ({\base + (\i-1)*\dx},0) {$X_{\t}$};
    }

    \foreach \i in {1} {
        \pgfmathtruncatemacro{\p}{(\j-1)*\ell + \i}
        \node[localpar] at ({\base + (\r+\i-1)*\dx},0) {$L_{\p}$};
    }

    \node at ({\base + 1.5*\dx}, -1) {Local group-$\j$};
}

\foreach \i in {1,...,3} {
    \node[globalpar] at ({\totalwidth + (\i-1)*\dx}, 0) {$G^I_{\i}$};
}

\node at ({\totalwidth + \dx}, -1.3) {Initial Global Parities };

\foreach \j in {4,...,6} {

    \pgfmathsetmacro{\base}{(\j-4)*(\r*\dx + \ell*\dx + \gap)}

    \foreach \i in {1,...,3} {
        \pgfmathtruncatemacro{\t}{9 + (\j-4)*\r + \i}
        \node[info] at ({\base + (\i-1)*\dx}, \yshift) {$X_{\t}$};
    }

    \foreach \i in {1} {
        \pgfmathtruncatemacro{\p}{\ell*\groups + (\j-4)*\ell + \i}
        \node[localpar] at ({\base + (\r+\i-1)*\dx}, \yshift) {$L_{\p}$};
    }

    \node at ({\base + 1.5*\dx}, \yshift -1) {Local group-$\j$};
}

\foreach \i in {4,...,6} {
    \node[globalpar] at ({\totalwidth + (\i-4)*\dx}, \yshift) {$G^I_{\i}$};
}

\node[align=right] at (-2, 3*\yshift - 0.55 ) {\textbf{Final codeword}};

\foreach \j in {1,...,3} {

    \pgfmathsetmacro{\base}{(\j-1)*(\r*\dx + \ell*\dx + \gap)}

    \foreach \i in {1,...,3} {
        \pgfmathtruncatemacro{\t}{(\j-1)*\r + \i}
        \node[info] at ({\base + (\i-1)*\dx},3*\yshift) {$X_{\t}$};
    }

    \foreach \i in {1} {
        \pgfmathtruncatemacro{\p}{(\j-1)*\ell + \i}
        \node[localpar] at ({\base + (\r+\i-1)*\dx},3*\yshift) {$L_{\p}$};
    }
}

\foreach \j in {4,...,6} {

    \pgfmathsetmacro{\base}{(\j-4)*(\r*\dx + \ell*\dx + \gap)}

    \foreach \i in {1,...,3} {
        \pgfmathtruncatemacro{\t}{9 + (\j-4)*\r + \i}
        \node[info] at ({\base + (\i-1)*\dx}, 3*\yshift - 1.1) {$X_{\t}$};
    }

    \foreach \i in {1} {
        \pgfmathtruncatemacro{\p}{\ell*\groups + (\j-4)*\ell + \i}
        \node[localpar] at ({\base + (\r+\i-1)*\dx}, 3*\yshift - 1.1) {$L_{\p}$};
    }

}

\foreach \i in {1,...,4} {
    \node[globalpar] at ({\totalwidth + (\i-1)*\dx}, 3*\yshift - 0.55) {$G^F_{\i}$};
}

\node at ({\totalwidth + 1.5*\dx},3*\yshift - 1.3) {Final Global Parities };

\draw[thick, dashed, rounded corners]
(-4.2, 0.8) rectangle ({\totalwidth + 3*\dx + 0.5}, \yshift - 1.3);

\draw[thick, dashed, rounded corners]
(-3.5, 3*\yshift + 0.8) rectangle ({\totalwidth + 3.5*\dx + 0.2}, 3*\yshift - 1.8);

\draw[->, thick]
({\totalwidth/2}, 1.6*\yshift )
--
({\totalwidth/2}, 2.6*\yshift)
node[midway, right] {};

\end{tikzpicture}

\caption{Stable optimal-distance $(\Ik=9,\Ig=3,r=3,\ell=1;\Fk=18,\Fg=4,r=3,\ell=1)$-LRCC. We adopt a global indexing scheme for both the initial and final codes. Moreover, since the local groups remain unchanged, their labels are omitted in the final code for brevity.}
\label{fig:lrc_global_merge_stable}
\end{figure}

Next, we introduce additional notation, and outline the key components and assumptions. Let $t \in [\Il]$ be the index of an arbitrary initial codeword. For ease of exposition, we adopt global indexing for each of the following: initial codewords, information nodes, local parity nodes, local groups, initial global parity nodes and final global parity nodes (see Figure~\ref{fig:lrc_global_merge_stable}).

\noindent\textit{Information Nodes and Local Parity Nodes:} 
\begin{enumerate}[(i)]
    \item Let $j \in [\Ik]^t$ be the index of an arbitrary information node of the $t$-th initial codeword, and let $X_{j}$ denote the random variable corresponding to the data stored in the $j$-th information node. Recall that $[\Ik]^t = \{(t-1)\Ik + 1, \ldots, t\Ik\}$. We assume that the random variables $X_{[\Il\Ik]}$ are independent and uniformly distributed.
    \item Let $\tau \in [\Imm]^t$ be the index of an arbitrary local group of the $t$-th initial codeword, let $a \in [\ell]^{\tau}$ be the index of an arbitrary local parity node belonging to the $\tau$-th local group, and let $L_a$ denote the random variable corresponding to the data stored in the $a$-th local parity node. It follows that
    \begin{equation}
        \condH{L_{[\ell]^{\tau}}} 
    {X_{[r]^{\tau}}} = 0\;. 
    \end{equation}
    Moreover, each local group is independent of the others, so
    \begin{equation}
        \Iinfo{X_{[\Il\Ik] \setminus [r]^\tau}}{L_{[\ell]^\tau}} = 0\;.
    \end{equation}
    \item Let $\mathcal{A} \subseteq [\ell]^\tau$ and $\mathcal{D} \subseteq [r]^\tau$ be arbitrary subsets of indices of local parity nodes and information nodes, respectively, of the $\tau$-th local group such that $|\mathcal{A}| + |\mathcal{D}| = r$. From Definition~\ref{def:locality}, it follows that 
    \begin{equation}
        \condH{X_{[r]^\tau}}{L_{\mathcal{A}}, X_{\mathcal{D}}}= 0\;.
    \end{equation}
\end{enumerate}

\noindent\textit{Initial Global Parity Nodes:}
\begin{enumerate}[(i)]
    \item Let $i \in [\Ig]^t$ be the index of an arbitrary global parity node of the $t$-th initial codeword, and let $G_i^I$ denote the random variable corresponding to the data stored in  the $i$-th initial global parity node. It follows that
    \begin{equation}
        \condH{G_{[\Ig]^t}^I}{X_{[\Ik]^t}} = 0\;.
    \end{equation}
    Additionally, each initial codeword is independent of the others, so
    \begin{equation}
        \Iinfo{X_{[\Il\Ik] \setminus [\Ik]^t}}{G^I_{[\Ig]^t}} = 0\;.
    \end{equation}
\end{enumerate}

\noindent\textit{Final Global Parity Nodes:}
\begin{enumerate}[(i)]
    \item Let $i \in [\Fg]$ be the index of an arbitrary final global parity node, and let $G_i^F$ denote the random variable corresponding to the data stored in the $i$-th final global parity node. It follows that
    \begin{equation}
        \condH{G_{[\Fg]}^F}{X_{[\Fk]}} = 0\;.
    \end{equation}
\end{enumerate}

Lastly, we assume that all nodes have a storage capacity of $\alpha$ symbols over the finite field $\mathbb{F}_q$.

\subsection{Conversion Coordinator and Bandwidth Analysis}

Recall that for stable codes, it suffices to focus solely on the read bandwidth cost, since the write bandwidth cost is fixed.

To model the conversion process, we introduce the following random variables (all globally indexed):
\begin{enumerate}[(i)]
    \item $V_j$: the data downloaded from the $j$-th information node, where $j \in [\Il\Ik]$.
    \item $U_i$: the data downloaded from the $i$-th initial global parity node, where $i \in [\Il\Ig]$.
    \item $W_a$: the data downloaded from the $a$-th local parity node, where $a \in [\Il\Imm\ell]$.
\end{enumerate}
The following properties hold for the random variables involved in the conversion procedure:
\begin{enumerate}[(i)]
    \item Each random variable corresponding to the downloaded data is a deterministic function of its stored data. Thus,
    \begin{equation}
        \condH{V_j}{X_j} = \condH{U_i}{G^I_i} = \condH{W_a}{L_a} = 0\;.
    \end{equation}
    \item Let the number of symbols downloaded from the $j$-th information node, $i$-th initial global parity node  and $a$-th local parity node be denoted by $\beta_j, \sigma_i$ and $\delta_a$, respectively. It follows that
    \begin{equation}\label{eq:downloaded_data_bounds}
    \Hent{V_j}\leq\beta_j\leq\alpha,\quad\Hent{U_i}\leq\sigma_i\leq\alpha,\quad\Hent{W_a}\leq\delta_a\leq\alpha\;.
    \end{equation}
 \item \textbf{Conversion Coordinator Property:} The final parity nodes are deterministically generated by the conversion coordinator, hence it follows that 
 \begin{align}\label{eq:conv_coord_prop}
 \condH{G_{[\Fg]}^F}{V_{[\Il\Ik]}, U_{[\Il\Ig]}, W_{[\Il\Imm\ell]}} \;=\; 0\;.
 \end{align}
 
\end{enumerate}

\noindent Note that the read bandwidth cost, denoted by $\bw$, is the total number of symbols downloaded by the conversion coordinator during the conversion process. In other words, the read bandwidth cost of any $(\Ik, \Ig, r, \ell; \Fk = \Il \Ik, \Fg, r, \ell)$-LRCC is defined as
\begin{equation}\label{eq:read_bandwidth_definition}
    \bw\;:=\; \sum_{j=1}^{\Il\Ik} \beta_j + \sum_{i=1}^{\Il\Ig} \sigma_i + \sum_{a = 1}^{\Il\Imm\ell}\delta_a\;.
\end{equation}

We first derive a constraint that the entropy of the data downloaded from each of the information nodes, initial global parity nodes and local parity nodes must satisfy. In turn, this constraint will be utilized to derive a lower bound on the read bandwidth cost of {stable optimal-distance} LRCCs.

\begin{lemma}\label{lem:bandwidth_constraint}
    For any stable optimal-distance $(\Ik, \Ig, r, \ell; \Fk = \Il \Ik, \Fg, r, \ell)$-LRCC, it holds that
    \begin{align*}
        &\sum_{t\in[\Il]} \Hent{U_{[\Ig]^t}}+\frac{\min\{\Fg,r\}+\ell}{\Imm(r+\ell)}\left[\sum_{j\in[\Il\Ik]}\Hent{V_j} +\sum_{a\in[\Il\Imm\ell]} \Hent{W_a}\right]\;\geq\; \Il\min\{\Fg,r\}\alpha\;.
    \end{align*}
\end{lemma}

\begin{IEEEproof}
 Let $t\in[\Il]$ be the index of an arbitrary initial codeword and $\tau \in [\Imm]^t$ be the index of an arbitrary local group of the $t$-th initial codeword. Then,
\begin{align*}
   0 &= \condH{G_{[\Fg]}^F}{V_{[\Il\Ik]}, U_{[\Il\Ig]}, W_{[\Il\Imm\ell]}}\tag{Equation~\eqref{eq:conv_coord_prop}: Conversion Coordinator Property}\\
   &\geq \condH{G_{[\Fg]}^F}{V_{[\Il\Ik]}, U_{[\Il\Ig]\setminus[\Ig]^t}, W_{[\Il\Imm\ell]}} - \Hent{U_{[\Ig]^t}}\\
    &= \condH{G_{[\Fg]}^F}{V_{[r]^\tau}, W_{[\ell]^\tau}, V_{[\Il\Ik]\setminus[r]^\tau}, U_{[\Il\Ig]\setminus[\Ig]^t}, W_{[\Il\Imm\ell]\setminus[\ell]^\tau}} - \Hent{U_{[\Ig]^t}}\\
    &\geq \condH{G_{[\Fg]}^F}{V_{[r]^\tau}, W_{[\ell]^\tau}, X_{[\Il\Ik]\setminus[r]^\tau}} - \Hent{U_{[\Ig]^t}}\tag{Data Processing Inequality}\\
    &\geq \condH{G_{[\min\{\Fg,r\}]}^F}{V_{[r]^\tau}, W_{[\ell]^\tau}, X_{[\Il\Ik]\setminus[r]^\tau}} - \Hent{U_{[\Ig]^t}}\\
    &= \condH{G_{[\min\{\Fg,r\}]}^F}{X_{[\Il\Ik]\setminus[r]^\tau}}  - \condI{G_{[\min\{\Fg,r\}]}^F}{V_{[r]^\tau}, W_{[\ell]^\tau}}{X_{[\Il\Ik]\setminus[r]^\tau}}- \Hent{U_{[\Ig]^t}}\\
    &= \min\{\Fg,r\}\alpha- \condI{G_{[\min\{\Fg,r\}]}^F}{V_{[r]^\tau}, W_{[\ell]^\tau}}{X_{[\Il\Ik]\setminus[r]^\tau}} - \Hent{U_{[\Ig]^t}}\;.\tag{Proposition~\ref{prop:independence}}
    \end{align*}
Let $\cA$ denote the set of all subsets of $V_{[r]^\tau} \cup W_{[\ell]^\tau}$ of size $\min\{\Fg,r\} + \ell$. For any subset of random variables $\cZ \in \cA$, from Proposition~\ref{prop:independence}, $G_{[\min\{\Fg,r\}]}^F$ and $\left(V_{[r]^\tau} \cup W_{[\ell]^\tau}\right) \setminus \cZ$ are conditionally independent with respect to $X_{[\Il\Ik]\setminus[r]^\tau}$. Then, it follows that
\begin{align*}
    0 &\geq \min\{\Fg,r\}\alpha- \min_{Z \in \cA}\left[\condH{Z}{X_{[\Il\Ik]\setminus[r]^\tau}} \right]- \Hent{U_{[\Ig]^t}}\tag{Lemma~\ref{lem:conditional_mutual_information_bound}}\\
    &= \min\{\Fg,r\}\alpha- \min_{Z \in \cA}\left[\Hent{Z} \right]- \Hent{U_{[\Ig]^t}}\\
    &\geq \min\{\Fg,r\}\alpha
    - \frac{\min\{\Fg,r\} + \ell}{r + \ell}\left[\sum_{j \in [r]^\tau}\Hent{V_j} + \sum_{a \in [\ell]^\tau}\Hent{W_a}\right]- \Hent{U_{[\Ig]^t}}\;.\tag{Lemma~\ref{lem:min_avg_entropy_bound}}
\end{align*}
Upon summing over $\tau \in [\Imm]^t$ and $t \in [\Il]$ and simple algebraic manipulation, we obtain the desired bound.
\end{IEEEproof}

Finally, we use the previous constraint to derive a lower bound on the read bandwidth cost of {stable optimal-distance} LRCCs in the following theorem.

\begin{theorem}\label{thm:optimalLRC}
The read bandwidth cost of any {stable optimal-distance} $(\Ik, \Ig, r, \ell; \Fk = \Il \Ik, \Fg, r, \ell)$-LRCC satisfies: 
\begin{equation*}
    \bw \geq \begin{cases}
         \Il r\alpha & \min\{\Ig, \Fg\} > r\;,\\
        \Il\Fg\alpha & \Ig \geq \Fg,\ r \ge \Fg\;,\\
        \Il\Ig\alpha + \Il\Imm(\Fg-\Ig)\left(\frac{r+\ell}{\Fg+\ell}\right)\alpha&\Ig < \Fg \leq r\;,\\
        \Il\Ig\alpha + \Il\Imm(r-\Ig)\alpha & \text{~otherwise}\;.
    \end{cases}
\end{equation*}
Moreover, this bound is tight when $\Fg \leq r$.
\end{theorem}

\begin{IEEEproof}
    Observe that 
    \begin{align*}
        \bw\;=\;&\sum_{j\in[\Il\Ik]} \beta_j + \sum_{i\in[\Il\Ig]} \sigma_i + \sum_{a\in[\Il\Imm\ell]}\delta_a\tag{Equation~\eqref{eq:read_bandwidth_definition}}\\
        \;\geq\;&\sum_{j\in[\Il\Ik]} \Hent{V_j} + \sum_{t\in[\Il]} \Hent{U_{[\Ig]^t}} + \sum_{a\in[\Il\Imm\ell]}\Hent{W_a}\tag{Equation~\eqref{eq:downloaded_data_bounds}}\\
        \;=\;&\frac{\Imm(r+\ell)}{\min\{\Fg,r\}+\ell}\sum_{t\in[\Il]} \Hent{U_{[\Ig]^t}}- \frac{\Imm(r+\ell)}{\min\{\Fg,r\}+\ell}\sum_{t\in[\Il]} \Hent{U_{[\Ig]^t}} + \sum_{t\in[\Il]} \Hent{U_{[\Ig]^t}}\\
        &\quad\; + \sum_{j\in[\Il\Ik]}\Hent{V_j} + \sum_{a\in[\Il\Imm\ell]} \Hent{W_a}\\
        \;=\;&\frac{\Imm(r+\ell)}{\min\{\Fg,r\}+\ell}\sum_{t\in[\Il]} \Hent{U_{[\Ig]^t}}- \frac{\Imm(r+\ell)}{\min\{\Fg,r\}+\ell}\sum_{t\in[\Il]} \Hent{U_{[\Ig]^t}} + \sum_{t\in[\Il]} \Hent{U_{[\Ig]^t}}\\
        &\quad\; + \frac{\Imm(r+\ell)}{\min\{\Fg,r\}+\ell}\left[ \frac{\min\{\Fg,r\}+\ell}{\Imm(r+\ell)}\left[\sum_{j\in[\Il\Ik]}\Hent{V_j} + \sum_{a\in[\Il\Imm\ell]} \Hent{W_a} \right]\right]\\
        \;=\;&\frac{\Imm(r+\ell)}{\min\{\Fg,r\}+\ell}\left[\sum_{t\in[\Il]} \Hent{U_{[\Ig]^t}} + \frac{\min\{\Fg,r\}+\ell}{\Imm(r+\ell)}\left[\sum_{j\in[\Il\Ik]}\Hent{V_j} + \sum_{a\in[\Il\Imm\ell]} \Hent{W_a} \right]\right]\\
        &\quad\;- \left[\frac{\Imm(r+\ell)}{\min\{\Fg,r\}+\ell}-1\right]\sum_{t\in[\Il]} \Hent{U_{[\Ig]^t}}\\
        \;\geq\;&\left[\frac{\Imm(r+\ell)}{\min\{\Fg,r\}+\ell}\right]\Il\min\{\Fg,r\}\alpha- \left[\frac{\Imm(r+\ell)}{\min\{\Fg,r\}+\ell}-1\right]\Il\Ig\alpha\tag{Lemma~\ref{lem:bandwidth_constraint}, Equation~\eqref{eq:downloaded_data_bounds}}\\
        \;=\;&\Il\min\{\Fg,r\}\alpha + \left[\frac{\Imm(r+\ell)}{(\min\{\Fg,r\}+\ell)}-1\right][\Il\min\{\Fg,r\}\alpha - \Il\Ig\alpha]\;.
    \end{align*}
    Note also that
    \begin{align*}
        \bw\;=\;&\sum_{j\in[\Il\Ik]} \beta_j + \sum_{i\in[\Il\Ig]} \sigma_i + \sum_{a\in[\Il\Imm\ell]}\delta_a\tag{Equation~\eqref{eq:read_bandwidth_definition}}\\
        \;\geq\;&\sum_{j\in[\Il\Ik]} \Hent{V_j} + \sum_{t\in[\Il]} \Hent{U_{[\Ig]^t}} + \sum_{a\in[\Il\Imm\ell]}\Hent{W_a}\tag{Equation~\eqref{eq:downloaded_data_bounds}}\\
        \;\geq\;&\sum_{t\in[\Il]} \Hent{U_{[\Ig]^t}} + \frac{\min\{\Fg,r\}+\ell}{\Imm(r+\ell)}\left[\sum_{j\in[\Il\Ik]}\Hent{V_j} + \sum_{a\in[\Il\Imm\ell]} \Hent{W_a} \right]\\
        \;\geq\;&\Il\min\{\Fg,r\}\alpha\;.\tag{Lemma~\ref{lem:bandwidth_constraint}}
    \end{align*}
    We can combine these two bounds as follows:
    $$\bw\;\geq\;\Il\min\{\Fg,r\}\alpha + \left[\frac{\Imm(r+\ell)}{(\min\{\Fg,r\}+\ell)}-1\right][\Il\min\{\Fg,r\}\alpha - \Il\min\{\Ig,\Fg,r\}\alpha]\;.$$
    Equivalently, as desired,
    \begin{align*}
    \bw \geq \begin{cases}
         \Il r\alpha & \min\{\Ig, \Fg\} > r\;,\\
        \Il\Fg\alpha & \Ig \geq \Fg, r \ge \Fg\;,\\
        \Il\Ig\alpha + \Il\Imm(\Fg-\Ig)\left(\frac{r+\ell}{\Fg+\ell}\right)\alpha&\Ig < \Fg \leq r\;,\\
        \Il\Ig\alpha + \Il\Imm(r-\Ig)\alpha & \text{~otherwise}\;.
    \end{cases}
    \end{align*}
    Moreover, for the regime $\Fg \leq r$, this bound is tight as it matches the read bandwidth cost of the stable optimal-distance LRCC constructions presented in~\cite{maturana2023LRCC}.
\end{IEEEproof}

\section{Conclusion}\label{sec:conclusion}
In this work, we studied fundamental limits on the bandwidth cost of stable optimal-distance LRCCs in the global merge regime.  We established, to our knowledge, the first non-trivial lower bounds on bandwidth cost in this regime. We did so by generalizing the information-theoretic approach of~\cite{singhvi2025tightlowerboundsbandwidth} to the setting of LRCs. Our bounds are tight for the regime $\Fg \le r$ as they match the bandwidth cost of constructions presented in~\cite{maturana2023LRCC}.

\paragraph*{Open Problems and Future Work}
A natural extension of this work would be to  improve on either the derived bounds or the existing constructions presented in~\cite{maturana2023LRCC} for the regime $\Fg > r$. A second future avenue of work would be to utilize this information-theoretic approach to derive bounds on the bandwidth cost of LRCCs in the global split regime and determine whether the constructions presented in~\cite{maturana2023LRCC} are optimal. Another problem to investigate is how removing the restrictions of stability and disjoint local groups, or adding the restrictions of maximum reliability and availability, affects the bounds on bandwidth cost. A broader direction to approach would be to consider conversion between LRCs while changing locality parameters and develop the corresponding information-theoretic approach.
Another promising direction is to generalize the information-theoretic framerwork introduced in \cite{singhvi2025tightlowerboundsbandwidth} to settings beyond MDS codes and LRCs, including Reed–Muller codes and regenerating codes, and possibly to more general classes of codes. 
\newpage
\bibliographystyle{IEEEtran} 
\bibliography{ref}

\end{document}